\documentclass[sigplan,10pt,nonacm=true]{acmart}

\usepackage{hyperref}
\usepackage{subcaption}
\usepackage{xcolor,colortbl}
\usepackage[skins]{tcolorbox}
\usepackage{amsmath}
\usepackage{enumerate}
\usepackage{tikz}
\usetikzlibrary{automata,arrows}
\usepackage[page]{appendix}
\usepackage{amsthm}
\usepackage[normalem]{ulem}
\usepackage{xfrac}

\copyrightyear{2022} 
\acmYear{2022} 
\setcopyright{rightsretained} 
\acmConference[PaPoC'22]{9th Workshop on Principles and Practice of Consistency for Distributed Data }{April 5--8, 2022}{RENNES, France}
\acmBooktitle{9th Workshop on Principles and Practice of Consistency for Distributed Data (PaPoC'22), April 5--8, 2022, RENNES, France}\acmDOI{10.1145/3517209.3524040}
\acmISBN{978-1-4503-9256-3/22/04}

\begin{document}

\title{Relaxed Paxos: Quorum Intersection Revisited (Again)}

\author{Heidi Howard}
\email{heidi.howard@cl.cam.ac.uk}
\affiliation{%
  \institution{University of Cambridge}
  \city{}
  \state{}
  \country{}
}

\author{Richard Mortier}
\email{richard.mortier@cl.cam.ac.uk}
\affiliation{%
  \institution{University of Cambridge}
  \city{}
  \state{}
  \country{}
}

\begin{abstract}
Distributed consensus, the ability to reach agreement in the face of failures, is a fundamental primitive for constructing reliable distributed systems.
The Paxos algorithm is synonymous with consensus and widely utilized in production.
Paxos uses two phases: phase one and phase two, each requiring a quorum of acceptors, to reach consensus during a round of the protocol.
Traditionally, Paxos requires that all quorums, regardless of phase or round, intersect and majorities are often used for this purpose.
Flexible Paxos proved that it is only necessary for phase one quorum of a given round to intersect with the phase two quorums of all previous rounds.

In this paper, we re-examine how Paxos approaches the problem of consensus.
We look again at quorum intersection in Flexible Paxos and observe that quorum intersection can be safely weakened further.
Most notably, we observe that if a proposer learns that a value was proposed in some previous round then its phase one no longer needs to intersect with the phase two quorums from that round or from any previous rounds.
Furthermore, in order to provide an intuitive explanation of our results, we propose a novel abstraction for reasoning about Paxos which utilizes write-once registers.
\end{abstract}




\maketitle

\section{Introduction}
\label{sec:intro}

We depend upon distributed systems, yet the computers and networks that make up these systems are asynchronous and unreliable.
The longstanding problem of distributed consensus formalizes how to reliably reach agreement in such systems.
When solved, we become able to construct strongly consistent distributed systems from unreliable components~\cite{schneider_cs90}.

The Paxos algorithm~\cite{lamport_tcs98} is widely deployed in production to solve distributed consensus~\cite{burrows_osdi06,chandra_podc07,corbett_tocs13}.
%
%
Despite its popularity, Paxos is notoriously difficult to understand, leading to much follow up work, explaining the algorithm in simpler terms~\cite{lampson_wdag96,prisco_wdag97,lampson_podc01,lamport_sigact01,boichat_sigact03,guerraoui_cj07,meling_opodis13,ongaro_atc14,van_cs15}. 

Paxos operates over a series of \emph{rounds} (sometimes referred to as proposal numbers or ballot numbers). In each round, a \emph{proposer} may attempt to decide a value by executing a two phase protocol. Each phase of the protocol requires agreement from a quorum of \emph{acceptors}. All quorums are required to intersect, regardless of round or phase. In phase one of Paxos, the proposer learns which values have been \emph{accepted} in previous rounds. In phase two of Paxos, the proposer proposes a value in the current round. If the proposer learned in phase one that one or more values had been previously accepted then it must propose the value with the greatest round. If a quorum of acceptors accepts the proposal in phase two then the proposed value is decided.

Our understanding of Paxos is ever-evolving, for example, Flexible Paxos~\cite{howard_opodis16} observed that intersection is not needed between all quorums in Paxos. 
Specifically, the authors observed that it is only necessary for phase one quorums of a given round to intersect with the phase two quorums of all previous rounds.
If the same quorum system is used for all rounds, then this can be simplified to the statement that only phase one and phase two quorums must intersect.

This paper re-examines how Paxos approaches the problem of consensus with the aim of further weakening the quorum intersection requirements of Paxos and improving understanding of this famously difficult algorithm.
In recent years, immutability has been utilized in distributed systems to tame complexity~\cite{shapiro_report11,helland_acmq15}.

The success of these efforts has inspired us to apply immutability to the problem of distributed consensus.
We proceed as follows.
Once we have defined consensus~(\S\ref{sec:problem}), we propose an abstract solution to consensus that uses only write-once registers to enable more intuitive reasoning about safety~(\S\ref{sec:abstract}).
Using the abstractions developed so far, we then describe a novel variant of Paxos, known as \emph{Relaxed Paxos}~(\S\ref{sec:paxos}), which generalizes over Paxos and Flexible Paxos.
We compare our Relaxed Paxos protocol to the original Paxos algorithm ~(\S\ref{sec:quorums}) and observe the following: If a proposer learns that a value was proposed in some previous round then it is no longer required to hear from at least one acceptor in each quorum for both that round or from all previous rounds to complete phase one. 

\section{Problem definition}
\label{sec:problem}

The classic formulation of single-degree consensus considers how to decide upon a single value in a distributed system.
This seemingly simple problem is made non-trivial by the weak assumptions made about the underlying system: we assume only that the algorithm is correctly executed (i.e.,~the non-Byzantine model).
We do not assume that participants are reliable. For safety, we do not assume that the system is synchronous, participants may operate at arbitrary speeds and messages may be arbitrarily delayed.

We consider systems comprised of two types of participants: \emph{acceptors}, which store the value, and \emph{proposers}, which read/write the value.
Proposers take as input a value to be proposed and produce as output the value decided by the acceptors.
Messages may only be exchanged between proposers and acceptors and we assume that the set of participants, acceptors and proposers, is fixed and known to the proposers.

An algorithm, such as Paxos, solves consensus if it satisfies the following three requirements:
\begin{description}
  \item[Non-triviality] All output values must have been the input value of a proposer.
  \item[Agreement] All proposers that output a value must output the same value.
  \item[Progress] All proposers must eventually output a value.
\end{description}

As termination cannot be guaranteed in an asynchronous system where failures may occur~\cite{fischer_jacm85}, consensus algorithms need only guarantee progress assuming partial synchrony~\cite{dwork_jacm88}.



If we have only one acceptor, then solving consensus is straightforward.
Assume the acceptor has a single write-once register, $r_0$, to store the decided value.
A \emph{write-once register} is a persistent variable that once written cannot be modified.
Proposers send requests to the acceptor with their input value.
If $r_0$ is unwritten, the value received is written to $r_0$ and is returned to the proposer.
If $r_0$ is already written, then the value in register $r_0$ is read and returned to the proposer.
The proposer then outputs the returned value.
This algorithm achieves consensus but requires the acceptor to be available for proposers to terminate.
Overcoming this limitation requires the deployment of more than one acceptor, so we now consider how to generalize to multiple acceptors.

\section{Abstract solution to distributed consensus}
\label{sec:abstract}

Consider a finite set of acceptors, $\{a_0,a_1,\dots,a_n\}$, where each acceptor has an infinite series of write-once registers, $\{r_0,r_1,\dots\}$.
At any time, each register is in one of the three states:
\begin{itemize}
  \item \textbf{unwritten}, the starting state for all registers; 
  \item \textbf{contains a value}, e.g.,~A, B, C; or
  \item \textbf{contains \emph{nil}}, a special value denoted as $\bot$.
\end{itemize}

\begin{figure}
  \centering
      \begin{subfigure}[b]{0.45\textwidth}
    \centering
    \begin{tabular}{ |l|l|}
      \hline
      \boldmath$i$ & \boldmath$Q_i$\\
      \hline
      $0, 1,\dots$ & $\{\{a_0,a_1\}, \{a_0,a_2\}, \{a_1,a_2\}\}$ \\
      \hline
    \end{tabular}
    \caption{Majority quorums over 3 acceptors ($\{a_0,a_1,a_2\}$).}
    \label{fig:example_configs/3s_paxos}
    \end{subfigure}
    
    \begin{subfigure}[b]{0.45\textwidth}
      \centering
      \begin{tabular}{ |l|l| }
        \hline
        \boldmath$i$ & \boldmath$Q_i$ \\
        \hline
        $0$ & $\{\{a_0,a_1,a_2\}\}$ \\
        $1, 2, \dots $ & $\{\{a_0,a_1\}, \{a_0,a_2\}, \{a_1,a_2\}\}$ \\
        \hline
      \end{tabular}
      \caption{Quorums can vary by round. Round 0 uses all 3 acceptors, round 1 onwards uses majority quorums.}
      \label{fig:example_configs/3s}
    \end{subfigure}
    \begin{subfigure}[b]{0.45\textwidth}
    \centering
    \begin{tabular}{ |l|l|}
      \hline
      \boldmath$i$ & \boldmath$Q_i$ \\
      \hline
      $0, 2,\dots$ & $\{\{a_0,a_1\}\}$ \\
      $1, 3,\dots$ & $\{\{a_2,a_3\}\}$\\
      \hline
    \end{tabular}
    \caption{Quorums do not need to intersect. Even rounds use two acceptors and odd rounds using the other two acceptors.}
    \label{fig:example_configs/4s_disjoint}
    \end{subfigure}

  \caption{Sample quorum configurations.}
  \label{fig:example_configs}
\end{figure}

\begin{figure}
  \centering
    \begin{subfigure}[b]{0.23\textwidth}
      \centering
        \begin{tabular}{ |c |c c c|  }
          \hline
             & \boldmath{$a_0$} & \boldmath{$a_1$} & \boldmath{$a_2$} \\
          \hline
          \boldmath{$r_0$} & $\bot$ & $\bot$ & B\\
          \boldmath{$r_1$} & $\bot$ & $\bot$ & $\bot$\\
          \boldmath{$r_2$} &  & \cellcolor{gray!20}A & \cellcolor{gray!20}A \\
          
        \end{tabular}
      \subcaption{No decision in round 0 and 1. Value A decided in round $2$}
      \label{fig:example_state/3s/a}
      \end{subfigure}
      \begin{subfigure}[b]{0.23\textwidth}
        \centering
        \begin{tabular}{ |c |c c c|  }
          \hline
               & \boldmath{$a_0$} & \boldmath{$a_1$} & \boldmath{$a_2$} \\
            \hline
            \boldmath{$r_0$} & $\bot$ & \cellcolor{gray!20}A & \cellcolor{gray!20}A \\
            \boldmath{$r_1$} & \cellcolor{gray!20}A & \cellcolor{gray!20}A & \\
           
          \end{tabular}
        \subcaption{Value A decided in round $0$ and round $1$}
        \label{fig:example_state/3s/b}
        \end{subfigure}

  \caption{Sample state tables for a system using majority quorums (Figure~\ref{fig:example_configs/3s_paxos}).}
  \label{fig:example_state}
\end{figure}

A quorum, $Q$, is a non-empty subset of acceptors, such that if all acceptors have the same (non-nil) value $v$ in the same register $r_i$ then value $v$ is said to be \emph{decided}.
The state of each \emph{round}, $i \in \mathbb{N}_{0}$, is the set comprised of the register $r_i$ from each acceptor.
Each round $i$ is configured with a set of quorums, $\mathcal{Q}_i$, and some examples are given in Figure~\ref{fig:example_configs}.
The state of all registers across the acceptors can be represented in a table, known as a \emph{state table}, where each column represents the state of one acceptor and each row represents a register.
By combining a configuration with a state table, we can determine whether any decision(s) have been reached, as shown in Figure~\ref{fig:example_state}.

\begin{figure}
  \begin{tcolorbox}
  \textbf{Rule 1: Quorum agreement}. A proposer may only output a value $v$ if it has read $v$ from register $r_i$ on a quorum of acceptors $Q \in \mathcal{Q}_i$.

  \textbf{Rule 2: New value}. A proposer may only write a (non-nil) value $v$ provided that either $v$ is the proposer's input value or that the proposer has read $v$ from a register.
  
  \textbf{Rule 3: Current decision}. A proposer may only write a (non-nil) value $v$ to register $r_i$ provided that no value $v'$ where $v \neq v'$ can also be decided in round $i$.
  
  \textbf{Rule 4: Previous decisions}. A proposer may only write a (non-nil) value $v$ to register $r_i$ provided no value $v'$ where $v \neq v'$ can be decided in rounds $0$ to $i-1$.

  \end{tcolorbox}
  \caption{The four rules for correctness.}
  \label{fig:rules}
\end{figure}

Figure~\ref{fig:rules} describes an abstract solution to consensus by giving four rules governing how proposers interact with registers to ensure that the safety requirements (non-triviality and agreement) for consensus are satisfied. See Appendix~\ref{sec:proofone} for a proof.

Rule 1 (\emph{quorum agreement}) ensures that proposers only output values that have been decided.
Rule 2 (\emph{new value}) ensures that only proposer input values can be written to registers thus only proposer input values can be decided and output by proposers.
Rules 3 and 4 ensure that no two quorums can decide upon different values.
Rule 3 (\emph{current decision}) ensures that all decisions made in a round will be for the same value whilst Rule 4 (\emph{previous decisions}) ensures that all decisions made by different rounds are for the same value.

Note that none of the four rules restrict when a proposer can write \emph{nil} ($\bot$) to a register.
The \emph{nil} value ensures that proposers can always safely write to any register.
This allows proposers to block quorums from making decisions if the proposer is unable to utilize the quorum due to Rule 3 or 4.
The \emph{nil} value is not necessary for safety, however, we will see later on how algorithms such as Paxos can utilize \emph{nil} to satisfy the progress requirement of consensus (\S\ref{sec:paxos}).


Rules 1 and 2 are easy to implement, but Rules 3 and 4 require more careful treatment.


\subsection{Satisfying rule three}

We can satisfy Rule 3 (\emph{current decision}) if we require that all (non-nil) values written to a given round are the same.
This can be achieved by assigning rounds to proposers in a round-robin fashion (for instance, with three proposers, $p_0$ can write to round 0, 3, 6, .. and $p_1$ can write to round 1, 4, 7, .. and so on) and requiring that proposers write at most one (non-nil) value to each of their own rounds.
This approach ensures that at most one (non-nil) value is ever written to each round and therefore at most one value can be decided by each round. 

\subsection{Satisfying rule four}

Rule 4 (\emph{previous decisions}) requires proposers to ensure that, before writing a (non-nil) value, previous rounds cannot decide a different value.
This is trivially satisfied for round $0$, however, more work is required by proposers to satisfy this rule for subsequent rounds.

Assume each proposer maintains its own local copy of the state table.
Initially, each proposer's state table is empty as they have not yet learned anything regarding the state of the acceptors.
A proposer can populate its state tables by reading registers and storing the results in its copy of the state table.
Since the registers are persistent and write-once, if a register contains a value (nil or otherwise) then any reads will always remain valid.
Each proposer's state tables will therefore always contain a subset of the values from the state table.

From its local state table, each proposer can track whether decisions have been reached or could be reached by previous quorums, using a \emph{decision table}. 
At any given time, each quorum is in one of four decision states:
\begin{description}
  \item[\textsc{Any}:] Any value could be decided by this quorum.
  \item[\textsc{Maybe}~\boldmath{$v$}:] If this quorum reaches a decision, then value $v$ will be decided.
  \item[\textsc{Decided}~\boldmath{$v$}:] The value $v$ has been decided by this quorum; a final state.
  \item[\textsc{None}:] This quorum will not decide a value; a final state.
\end{description}

\begin{figure}
  \centering
  \begin{footnotesize}
  \begin{tikzpicture}[node distance=1.2cm,every state/.style={thick, fill=gray!10}]
    \node[state,initial] (any) {\textsc{Any}};
    \node[state, right of=any,xshift=2cm,align=center] (maybe) {\textsc{Maybe} \\ $v$};
    \node[state, right of=maybe,xshift=1cm,yshift=2cm] (no) {\textsc{None}};
    \node[state, below of=no,yshift=-2cm,align=center] (yes) {\textsc{Decided} \\ $v$};

    \draw[->,thick,above,align=center,pos=0.5] (any) edge node{$\exists i'\geq i: a[i']=v$} (maybe);

    \draw[->,thick,below left ,align=center,pos=0.75, bend right] (any) edge node{$\forall a\in Q: a[i]=v$} (yes);
    \draw[->,thick,below left,align=left,pos=0.9] (maybe) edge node{$\forall a\in Q: a[i]=v$} (yes);

    \draw[->,thick,above left ,align=center,pos=0.75, bend left] (any) edge node{$\exists a\in Q: a[i]=\bot$} (no);
    \draw[->,thick,below right,align=left,pos=0.5] (maybe) edge node{$\exists a\in Q: a[i]=\bot$ \\ $\exists i'>i, v' \neq v: a[i']=v'$} (no);

  \end{tikzpicture}
\end{footnotesize}
  \caption{Computing the decisions state for quorum $Q \in \mathcal{Q}_i$ over round $i$. We use $a[i]=v$ as shorthand for when a proposer has read value $v$ from register $r_i$ on acceptor $a$. }
  \label{fig:decision_state_machine}
\end{figure}
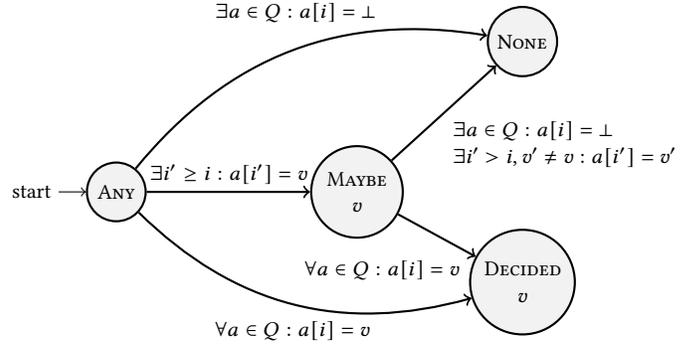

The rules for updating the decision table are detailed below and in Figure~\ref{fig:decision_state_machine}.
Initially, the decision state of all quorums is \textsc{Any}.
If there is a quorum where all registers contain the same (non-nil) value $v$ then its decision state is \textsc{Decided}~$v$.
When a proposer reads \emph{nil} from register $r_i$ on acceptor $a$ then for all quorums $Q \in \mathcal{Q}_i$ where $a \in Q$, the decision state \textsc{Any}/\textsc{Maybe}~$v$ becomes \textsc{None}.
When a proposer reads a non-\emph{nil} value $v$ from a register $r_i$ then for all quorums over rounds $0$ to $i$, the decision state \textsc{Any} becomes \textsc{Maybe}~$v$ and \textsc{Maybe}~$v'$ where $v\neq v'$ becomes \textsc{None}.
These rules use the knowledge that if a proposer reads a (non-nil) value $v$ from the register $r_i$ on acceptor $a$, it learns that all quorums in the round $i$ must decide $v$ if they reach a decision (Rule 3), and if any quorum of rounds $0$ to $i-1$ reaches a decision then value $v$ is decided (Rule 4).

\begin{figure}
  \begin{tcolorbox}
    \textbf{Rule 1: Quorum agreement}. A proposer may output value $v$ provided at least one quorum state is \textsc{Decided}~$v$.
    
    \textbf{Rule 2: New value}. A proposer $p$ may write a non-nil value $v$ to a register provided $v$ is $p$'s input value or has been read from a register.
    
    \textbf{Rule 3: Current decision}. A proposer $p$ may write a non-nil value $v$ to a register $r_i$ provided round $i$ has been allocated to $p$ but not yet used.
    
    \textbf{Rule 4: Previous decisions}. A proposer may write a non-nil value $v$ to register $r_i$ provided the decision state of each quorum from rounds $0$ to $i-1$ is \textsc{None}, \textsc{Maybe}~$v$ or \textsc{Decided}~$v$.

  \end{tcolorbox}
  \caption{The four rules for correctness using proposer decision tables.}
  \label{fig:decisiontable}
\end{figure}

Figure~\ref{fig:decisiontable} describes how proposers can use state tables and decision tables to implement all four rules for correctness (Fig.~\ref{fig:rules}).

\section{Relaxed Paxos}
\label{sec:paxos}

\begin{figure*}
  \centering
  \begin{tikzpicture}
    \draw [-] (0.5,1) node[left] {$p_0$} -- (3.5,1);
    \draw [-] (5,1) node[left] {$p_1$} -- (10,1);
    \draw [->] (-2,0.5) node[left] {$a_0$} -- (12,0.5);
    \draw [->] (-2,0) node[left] {$a_1$} -- (12,0);
    \draw [->] (-2,-0.5) node[left] {$a_2$} -- (12,-0.5);

    \draw [->,thick] (1,1) node[rotate=50,above,anchor=west] {$\langle\textsc{P2a},0,A\rangle$} -- (1.5,0.5);
    \draw [<-,thick] (2,1) node[rotate=50,above,anchor=west] {$\langle\textsc{P2b},0,A\rangle$} -- (1.5,0.5);
    \draw [->,thick] (1,1) -- (2,0);
    \draw [<-,thick] (3,1) node[rotate=50,above,anchor=west] {$\langle\textsc{P2b},0,A\rangle$} -- (2,0);

    \draw [->,thick] (1+4.5,1) node[rotate=50,above,anchor=west] {$\langle\textsc{P1a},1\rangle$} -- (1.5+4.5,0.5);
    \draw [<-,thick] (2+4.5,1) node[rotate=50,above,anchor=west] {$\langle\textsc{P1b},1,\{r_0:A\}\rangle$} -- (1.5+4.5,0.5);

    \draw [->,thick] (4+3.5,1) node[rotate=50,above,anchor=west] {$\langle\textsc{P2a},1,A\rangle$} -- (4.5+3.5,0.5);
    \draw [<-,thick] (5+3.5,1) node[rotate=50,above,anchor=west] {$\langle\textsc{P2b},1,A\rangle$} -- (4.5+3.5,0.5);
    \draw [->,thick] (4+3.5,1) -- (5+3.5,0);
    \draw [<-,thick] (6+3.5,1) node[rotate=50,above,anchor=west] {$\langle\textsc{P2b},1,A\rangle$} -- (5+3.5,0);
  \end{tikzpicture}
  \caption{Sample message exchange for Relaxed Paxos between two proposers ($\{p_0,p_1\}$) and three acceptors ($\{a_0,a_1,a_2\}$). Majority quorums (Figure ~\ref{fig:example_configs/3s_paxos}) are used. Some messages are omitted for simplicity. Corresponding examples of the proposer's state table and decision table are given in Figure~\ref{fig:paxos_decision} for proposer $p_0$ and Figure~\ref{fig:paxostwo_decision} for proposer $p_1$.}
  \label{fig:paxos_msd}
\end{figure*}
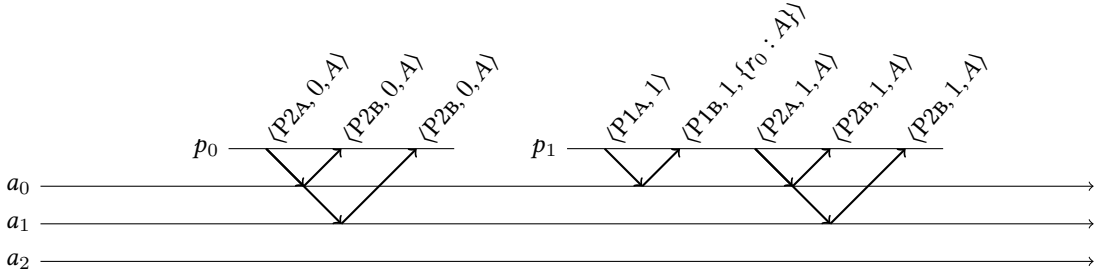
\begin{figure}
  \centering
  \begin{subfigure}[b]{\columnwidth}
    \centering
    \begin{tabular}{|c|ccc|}
      \hline
         & \boldmath{$a_0$} & \boldmath{$a_1$} & \boldmath{$a_2$} \\
      \hline
      \multicolumn{1}{|c|}{\boldmath{$r_0$}} &  &  &   \\
      
    \end{tabular}
    \hspace{0.5cm}
    \begin{tabular}{|l|l|l|}
      \hline
      \boldmath{$i$} & \boldmath{$Q$}& \textbf{Decision} \\
      \hline
      $0$ & $\{a_0,a_1\}$ & \textsc{Any} \\
      $0$ & $\{a_0,a_2\}$ & \textsc{Any} \\
      $0$ & $\{a_1,a_2\}$ & \textsc{Any} \\
 
    \end{tabular}
    \caption{Initial state.}
    \label{fig:paxos_decision/initial}
  \end{subfigure}

  \begin{subfigure}[b]{\columnwidth}
    \centering
    \begin{tabular}{|c|ccc|}
      \hline
         & \boldmath{$a_0$} & \boldmath{$a_1$} & \boldmath{$a_2$} \\
           \hline
      \multicolumn{1}{|c|}{\boldmath{$r_0$}} & A &  &  \\
      
    \end{tabular}
    \hspace{0.5cm}
    \begin{tabular}{|l|l|l|}
      \hline
      \boldmath{$i$} & \boldmath{$Q$}& \textbf{Decision} \\
      \hline
      $0$ & $\{a_0,a_1\}$ & \textsc{Maybe A} \\
      $0$ & $\{a_0,a_2\}$ & \textsc{Maybe A} \\
      $0$ & $\{a_1,a_2\}$ & \textsc{Maybe A} \\
  
    \end{tabular}
    \caption{State after receiving {$\langle\textsc{P2b},0,A\rangle$} from $a_0$.}
    \label{fig:paxos_decision/end}
  \end{subfigure}

  \begin{subfigure}[b]{\columnwidth}
    \centering
    \begin{tabular}{|c|ccc|}
      \hline
         & \boldmath{$a_0$} & \boldmath{$a_1$} & \boldmath{$a_2$} \\
      \hline
      \multicolumn{1}{|c|}{\boldmath{$r_0$}} & \cellcolor{gray!20}A & \cellcolor{gray!20}A &  \\
     
    \end{tabular}
    \hspace{0.5cm}
    \begin{tabular}{|l|l|l|}
      \hline
      \boldmath{$i$} & \boldmath{$Q$}& \textbf{Decision} \\
      \hline
      \cellcolor{gray!20}$0$ & \cellcolor{gray!20}$\{a_0,a_1\}$ & \cellcolor{gray!20}\textsc{Decided A} \\
      $0$ & $\{a_0,a_2\}$ & \textsc{Maybe A} \\
      $0$ & $\{a_1,a_2\}$ & \textsc{Maybe A} \\
      
    \end{tabular}
    \caption{State after receiving {$\langle\textsc{P2b},0,A\rangle$} from $a_1$.}
    \label{fig:paxos_decision/endtwo}
  \end{subfigure}
  
  \caption{Sample proposer state tables (left) and decision tables (right) for proposer $p_0$ during the execution in Figure~\ref{fig:paxos_msd}.}
  \label{fig:paxos_decision}
\end{figure}

\begin{figure}
  \centering
   \begin{subfigure}[b]{\columnwidth}
     \centering
        \begin{tabular}{| c |c c c | }
          \hline
             & \boldmath{$a_0$} & \boldmath{$a_1$} & \boldmath{$a_2$} \\
          \hline
          \multicolumn{1}{|c|}{\boldmath{$r_0$}} &  & &  \\
         
        \end{tabular}
        \hspace{0.5cm}
        \begin{tabular}{ |l| l |l|}
          \hline
      \boldmath{$i$} & \boldmath{$Q$}& \textbf{Decision} \\
          \hline
          $0$ & $\{a_0,a_1\}$ & \textsc{Any} \\
          $0$ & $\{a_0,a_2\}$ & \textsc{Any} \\
          $0$ & $\{a_1,a_2\}$ & \textsc{Any} \\
         
        \end{tabular}
     \caption{Initial state.}
     \label{fig:paxostwo_decision/initial}
   \end{subfigure}

  \begin{subfigure}[b]{\columnwidth}
    \centering
    \begin{tabular}{| c |c c c | }
      \hline
         & \boldmath{$a_0$} & \boldmath{$a_1$} & \boldmath{$a_2$} \\
      \hline
      \multicolumn{1}{|c|}{\boldmath{$r_0$}} & A & &  \\
     
    \end{tabular}
    \hspace{0.5cm}
    \begin{tabular}{ |l |l| l|}
      \hline
      \boldmath{$i$} & \boldmath{$Q$}& \textbf{Decision} \\
      \hline
      $0$ & $\{a_0,a_1\}$ & \textsc{Maybe A} \\
      $0$ & $\{a_0,a_2\}$ & \textsc{Maybe A} \\
      $0$ & $\{a_1,a_2\}$ & \textsc{Maybe A} \\
    
    \end{tabular}
    \caption{State after receiving $\langle\textsc{P1b},1,\{r_0:A\}\rangle$ from $a_0$.}
    \label{fig:paxostwo_decision/midone}
  \end{subfigure}


  \begin{subfigure}[b]{\columnwidth}
     \centering
     \begin{tabular}{ c |c c c | }
      \cline{2-4}
          & \boldmath{$a_0$} & \boldmath{$a_1$} & \boldmath{$a_2$} \\
      \hline
      \multicolumn{1}{|c|}{\boldmath{$r_0$}} & A &  &  \\
      \multicolumn{1}{|c|}{\boldmath{$r_1$}} & A &  &  \\
 
     \end{tabular}
     \hspace{0.5cm}
     \begin{tabular}{ |l l  l |}
      \hline
      \boldmath{$i$} & \boldmath{$Q$}& \textbf{Decision} \\
      \hline
      $0$ & $\{a_0,a_1\}$ & \textsc{Maybe A} \\
      $0$ & $\{a_0,a_2\}$ & \textsc{Maybe A} \\
      $0$ & $\{a_1,a_2\}$ & \textsc{Maybe A} \\
      \hline
      $1$ & $\{a_0,a_1\}$ & \textsc{Maybe A} \\
      $1$ & $\{a_0,a_2\}$ & \textsc{Maybe A} \\
      $1$ & $\{a_1,a_2\}$ & \textsc{Maybe A} \\

     \end{tabular}
     \caption{State after receiving $\langle\textsc{P2b},1,A\rangle$ from $a_0$.}
     \label{fig:paxostwo_decision/end}
  \end{subfigure}
   
     \begin{subfigure}[b]{\columnwidth}
     \centering
     \begin{tabular}{ c |c c c | }
       \cline{2-4}
          & \boldmath{$a_0$} & \boldmath{$a_1$} & \boldmath{$a_2$} \\
       \hline
       \multicolumn{1}{|c|}{\boldmath{$r_0$}} & A &  &  \\
       \multicolumn{1}{|c|}{\boldmath{$r_1$}} & \cellcolor{gray!20}A & \cellcolor{gray!20}A &  \\
       
     \end{tabular}
     \hspace{0.5cm}
     \begin{tabular}{ |l |l | l |}
       \hline
      \boldmath{$i$} & \boldmath{$Q$}& \textbf{Decision} \\
       \hline
       $0$ & $\{a_0,a_1\}$ & \textsc{Maybe A} \\
       $0$& $\{a_0,a_2\}$ & \textsc{Maybe A} \\
       $0$ & $\{a_1,a_2\}$ & \textsc{Maybe A} \\
       \hline
       \cellcolor{gray!20}$1$ & \cellcolor{gray!20}$\{a_0,a_1\}$ & \cellcolor{gray!20}\textsc{Decided A} \\
       $1$ & $\{a_0,a_2\}$ & \textsc{Maybe A} \\
       $1$ & $\{a_1,a_2\}$ & \textsc{Maybe A} \\
       
     \end{tabular}
     \caption{State after receiving $\langle\textsc{P2b},1,A\rangle$ from $a_1$.}
     \label{fig:paxostwo_decision/endtwo}
   \end{subfigure}
   
  \caption{Sample proposer state tables (left) and decision tables (right) for proposer $p_1$ during the execution in Figure~\ref{fig:paxos_msd}.}
  \label{fig:paxostwo_decision}
\end{figure}

We now describe Paxos~\cite{lamport_tcs98} in terms of decision tables.
We refer to our description of Paxos as \emph{Relaxed Paxos} to distinguish it from the usual descriptions of Paxos.
Like Paxos, Relaxed Paxos consists of two phases: phase one and phase two.
In phase one, after choosing a round $i$, a proposer reads from rounds $0$ to $i-1$ until it learns which value $v$ is safe to write to round $i$.
In phase two, the proposer writes the value $v$ to round $i$ and it outputs $v$ once it learns that a quorum of acceptors has value $v$ in register $r_i$.

We now consider each phase of Relaxed Paxos in more detail. 
Note that a proposer can complete phase one as soon as the completion criteria (underlined) has been satisfied.

\subparagraph*{Phase One}
\begin{itemize}
  \item A proposer $p$ chooses its next round $i$. \uline{Once the decision state of all quorums from rounds $0$ to $i-1$ is \textsc{None} or \textsc{Maybe}~$v$ then the proposer $p$ chooses the value $v$ (or if all states are \textsc{None} then its input value) and proceeds to phase two.}
  \item The proposer $p$ sends $\langle\textsc{P1a},i\rangle$ to all acceptors.
  \item Upon receiving $\langle\textsc{P1a},i\rangle$, each acceptor checks if register $r_i$ is unwritten.
  If so, any unwritten registers up to $r_{i-1}$ (inclusive) are set to \emph{nil}.
  The acceptor replies with $\langle\textsc{P1b},i,\mathcal{R}\rangle$ where $\mathcal{R}$ is a set of all written registers.
  \item Each time the proposer $p$ receives a \textsc{P1b}, it updates its state and decision tables accordingly.
  If the proposer $p$ times out before completing phase one, it restarts phase one with a greater round.
\end{itemize}

\subparagraph*{Phase Two}
\begin{itemize}
  \item The proposer $p$ sends $\langle\textsc{P2a},i,v\rangle$ to all acceptors where $i$ is the round chosen at the start of phase one and $v$ is the value chosen at the end of phase one.
  \item Upon receiving $\langle\textsc{P2a},i,v\rangle$, each acceptor checks if register $r_i$ is unwritten.
  If so, any unwritten registers up to $r_{i-1}$ (inclusive) are set to \emph{nil} and register $r_i$ is set to the value $v$.
  The acceptor replies with $\langle\textsc{P2b},i,v\rangle$.
  \item Each time the proposer $p$ receives a \textsc{P2a}, it updates its state and decision tables accordingly.
  Once the decision state of a quorum is \textsc{Decided}~$v$ then the proposer $p$ outputs the value $v$.
  If the proposer $p$ times out before completing phase two, it restarts phase one with a greater round.
\end{itemize}

Once a proposer outputs the decided value, its state and decision tables are no longer needed. 
Optionally, a \textsc{P3a} can be used to notify the acceptors of the decided value, which can be recorded and the registers garbage collected. 

Note that when a proposer restarts Relaxed Paxos after timing out, it does not need to clear its state table and decision table. 
Instead, the proposer can safely reuse what it learned about the state of acceptors in earlier executions. 

Similarly, the proposer can also update its decision table for its previously assigned rounds. For instance, if it did not write a value in an assigned round then the decision state is for all quorums in that round is \textsc{None}. 

Phase two implements Rule 1 by requiring a proposer to wait until it has a quorum in its decision table with decision state \textsc{Decided}~$v$ before it outputs $v$.
Relaxed Paxos ensures Rule 2 as a proposer will only write a (non-nil) value $v$ if $v$ is from a decision state \textsc{Maybe}~$v$ or if $v$ is the proposer's input value.
Each proposer will only execute Relaxed Paxos at most once for each assigned round, thus ensuring Rule 3.
The purpose of phase one is to implement Rule 4, by requiring that a proposer does not write any (non-nil) value $v$ to round $i$ until it has checked that all previous quorums (over rounds $0$ to $i-1$) either cannot reach a decision or will decide the same value $v$.
Before an acceptor sends a \textsc{P1b} in round $i$, it writes \emph{nil} to any unwritten registers from $r_0$ to $r_{i-1}$, effectively blocking previous rounds from deciding new values.

Figure~\ref{fig:paxos_msd} gives an example of the message exchange as two proposers execute Relaxed Paxos with three acceptors.

\section{Implications for Paxos}
\label{sec:quorums}

Relaxed Paxos differs from the usual descriptions of Paxos as it encapsulates various generalizations regarding quorums.

Relaxed Paxos can be configured with any mapping of quorums to rounds.
However, the choice of quorum configuration will impact the conditions under which progress can be guaranteed.
Paxos requires proposers to wait for responses from a quorum of acceptors in each phase of the algorithm and typically utilizes majority quorums as it requires all quorums to intersect, regardless of the round or phase of the algorithm.
Agreement from a majority of acceptors is therefore both necessary and sufficient for a proposer to complete phase one.

Typically, descriptions of Paxos require acceptors to maintain two variables: the last round promised and the last accepted proposal (consisting of a round number and value).
Relaxed Paxos instead uses a set of write-once registers to store accepted values.
The nil value is used to implement phase one, without the need for a secondary set of registers.

\subsection{Known results}

Paxos uses the same quorums system for all phases and rounds.
Instead, Flexible Paxos differentiates between the quorums used for each round and which phase of Paxos the quorum is used for.
$\mathcal{Q}_r^k$ is the set of quorums for phase $k$ of the round $i$.
Flexible Paxos observed that quorum intersection is required only between the phase one quorum for round $i$ and the phase two quorums of rounds $0$ to $i-1$.
More formally, $\forall i \in \mathbb{N}_{0}, \forall i' \in \mathbb{N}_{<i}: \mathcal{Q}_i^1 \cap \mathcal{Q}_{i'}^2 \neq \emptyset$.

We can observe the same result in Relaxed Paxos.
A proposer can always safely proceed to phase two of round $i$ after receiving \textsc{P1b} messages from at least one acceptor in each quorum from rounds $0$ to $i-1$.
This is because, upon receiving $\langle\textsc{P1b},i,\mathcal{R}\rangle$ from acceptor $a$, the proposer learns the contents of registers $r_0$ to $r_{i-1}$ on acceptor $a$ as all these registers will have already been written.
It is therefore the case that, once the proposer has updated its decision table, none of the quorums over rounds $0$ to $i-1$ which contain $a$ will have the decision state of \textsc{Any}.
It is also that case that the decision table will never contain \textsc{Maybe}~$v$ and \textsc{Maybe}~$v'$ for two different values $v \neq v'$ so the proposer will always be able to choose a single safe value to write.

One implication of this result is that no \textsc{P1b} messages are required to complete phase one for round $0$.
This is illustrated in Figure~\ref{fig:paxos_decision/initial} where the proposer $p_0$ was safe to proceed directly to phase two from startup.


\subsection{New results}

Relaxed Paxos further improves over Flexible Paxos as it also allows the proposer to safely proceed to phase two before hearing from a phase one quorums of acceptors.

We observe that it is possible for a proposer to proceed to phase two of round $i$ without receiving a \textsc{P1b} messages from at least one acceptor in each quorum from rounds $0$ to $i-1$.
If a proposer learns that a register $r_i$ contains a (non-nil) value $v$ then it also learns that if any quorums from rounds $0$ to $i$ reach a decision then $v$ must be chosen.
By updating their decision table, we observe that it is no longer necessary for the proposer in phase one to intersect with the phase two quorums of registers up to $r_i$ (inclusive).
This is illustrated in Figure~\ref{fig:paxostwo_decision/midone} where the proposer could safely proceed to phase two after one \textsc{P1b} message as the proposer reads a non-nil value from the predecessor round.

We also observe that the value selection rule (the condition that governs which values a proposer can safely write in phase two) is weaker in Relaxed Paxos than in the original Paxos protocol. Paxos permits a proposer to propose its input value in phase two only if it did not learn of any values in phase one. Relaxed Paxos also allows a proposer to propose its input value (or more generally, any value) if it knows that all values it has already learned cannot have been decided. This could be because another acceptor in each possible quorum has written nil to its register for that round.







\section{Conclusion}
\label{sec:conc}

Paxos has long been the \emph{de facto} approach to reaching consensus, however, it is also notoriously difficult to understand.
This work is the latest addition in a long line of papers that seek to explain Paxos in simpler terms, however, we are the first to extend Flexible Paxos and further relax the quorum intersection requirements of Paxos.
We have re-framed the problem of consensus in terms of write-once registers and presented an abstract solution to distributed consensus comprised of four rules which an algorithm must abide by in-order to correctly solve consensus.
Utilizing this abstraction, we have presented Relaxed Paxos, which further weakens Paxos's requirements for quorum intersection.

Further exploration of this result and our abstraction solution to consensus can be found in the author's thesis~\cite{howard_thesis} and the extended version of the paper on arxiv~\cite{howard_general19} respectively.



\section{Acknowledgements}
We would like to thank Jon Crowcroft, Stephen Dolan and Martin Kleppmann for their valuable feedback on earlier iterations of this paper.
This work was funded in part by EPSRC EP/R03351X/1 and EP/T022493/1.

\bibliographystyle{plainurl} 
\bibliography{refs}


\appendix
\section{Safety of the four rules for correctness}
\label{sec:proofone}

Figure~\ref{fig:rules} proposes four rules which we claim are sufficient to satisfy the safety (non-triviality and agreement) requirements of distributed consensus.
We now prove each requirement in turn.
We will use $a[i]=v$ to denote that register $r_i$ on acceptor $a$ contains the value $v$ and $a[i]=*$ to denote that register $r_i$ on acceptor $a$ is unwritten.

\begin{lemma}[Satisfying non-triviality]
If a value $v$ is the output of a proposer $p$ then $v$ was the input of some proposer $p'$.
\end{lemma}

\begin{proof}
Assume that the (non-nil) value $v$ was the output of proposer $p$.
According to Rule 1, at least one register must contain $v$.
Consider the invariant that all (non-nil) registers contain proposer input values.
Initially, all registers are unwritten thus this invariant holds.
According to Rule 2, each proposer will only write either their input value or a value copied from another register, thus the invariant will be preserved.
\end{proof}

\begin{lemma}[Satisfying agreement]
If two proposers, $p$ and $p'$, output values, $v$ and $v'$ (respectively), then $v=v'$.
\end{lemma}

\begin{proof}
Assume that the (non-nil) value $v$ was the output of proposer $p$ and that the (non-nil) value $v'$ was the output of proposer $p'$.
According to Rule 1, it must be the case that $\exists i \in \mathbb{N}_{0}, \exists Q \in \mathcal{Q}_i,\forall a \in Q: a[i]=v$ and $\exists i' \in \mathbb{N}_{0}, \exists Q' \in \mathcal{Q}_i,\forall a' \in Q': a'[i']=v'$.
Since rounds are totally ordered, it must be the case that either $i=i'$, $i < i'$ or $i > i'$.
We now consider each case in turn:
\begin{description}
  \item[Case] $i=i'$:\\
  According to Rule 3, a proposer will only write $v$ to round $i$ after ensuring no quorum in round $i$ can reach a different decision.
  Thus $v=v'$.
  \item[Case] $i < i'$:\\ According to Rule 4, a proposer will only write $v'$ to round $i'$ after ensuring no quorum in round $i$ can reach a different decision. Thus $v=v'$.
  \item[Case] $i > i'$:\\ This is the same as case $i < i'$ with $i$ and $i'$ swapped. Thus $v=v'$.
\end{description}
\end{proof}

\section{Safety of decision tables rules}
\label{subsec:abstract/correctness}

We have shown that the four rules for correctness are sufficient to satisfy the safety (non-triviality and agreement) requirements of consensus.
We will now show that the proposer decision table rules (Fig.~\ref{fig:decisiontable}) implement the four rules for correctness (Fig.~\ref{fig:rules}) and thus satisfies the safety requirements of consensus.

\begin{lemma}[Satisfying Rule 1]
If the value $v$ is the output of proposer $p$ then $p$ has read $v$ from a register $r_i$ on a quorum of acceptors $Q \in \mathcal{Q}_i$.
\end{lemma}

\begin{proof}
Assume the value $v$ is the output of the proposer $p$.
There must exist a round $i$ and quorum $Q \in \mathcal{Q}_i$ in the decision table of $p$ with the status \textsc{Decided}~$v$ (Fig.~\ref{fig:decisiontable}).
A quorum $Q$ can only reach decision state \textsc{Decided}~$v$ if $\forall a \in Q: a[i]=v$.
\end{proof}

\begin{lemma}[Satisfying Rule 2]
If the value $v$ is written by a proposer $p$ then either $v$ is $p$'s input value or $v$ has been read from some register.
\end{lemma}

\begin{proof}
Assume the value $v$ has been written by proposer $p$.
According to Figure~\ref{fig:decisiontable}, $v$ must be either the input value of $p$ or read from some register.
\end{proof}

\begin{lemma}[Satisfying Rule 3]
If the values $v$ and $v'$ are decided in round $i$ then $v=v'$.
\end{lemma}

\begin{proof}
Assume the values $v$ and $v'$ are decided in round $i$ and therefore there must be at least one acceptor $a$ where $a[i]=v$ and one acceptor $a'$ where $a[i]=v'$.
At most one proposer is assigned round $i$ and the proposer will only write one (non-nil) value to $i$ (Fig.~\ref{fig:decisiontable}) and so round $i$ will only contain one (non-nil) value thus $v=v'$.
\end{proof}

\begin{lemma}[Satisfying Rule 4]
If the value $v$ is decided in round $i$ and the (non-nil) value $v'$ is written to round $i'$ where $i<i'$ then $i=i'$
\end{lemma}

We will prove this by induction over the writes to rounds $>i$.
Note that this proof assumes only a partial order over writes as there may be concurrent writes to different acceptors.

\begin{lemma}[Satisfying Rule 4 - Base case]
If the value $v$ is decided in round $i$ then the first (non-nil) value to be written to a round $i'$ where $i<i'$ is $v$.
\end{lemma}

\begin{proof}
Assume the value $v$ is decided in round $i$ by quorum $Q \in \mathcal{Q}_i$ thus at some time $\forall a \in Q: a[i]=v$.
Since registers are write-once, it will always be that case that $\forall a \in Q: a[i]=v \lor *$.

Assume the value $v'$ is written to round $i'$ by proposer $p$ where $i<i'$.
Assume that $i'$ is the first non-nil value to be written to any register $>i$ thus $p$ cannot read any (non-nil) values from registers $>i$ before writing $v'$. 
We will show that $v=v'$.

Consider the decision table of proposer $p$ when it is writing $v'$ to $i'$.
Since $i<i'$, the decision state of $Q$ must be either \textsc{None}, \textsc{Maybe}~$v'$ or \textsc{Decided}~$v'$ (Fig.~\ref{fig:decisiontable}).
We now consider each case in turn.
\begin{description}
  \item[Case] \textsc{Decided} $v'$:\\
  This decision state requires that $\forall a \in Q: a[i]=v'$.
  Since we know that $\forall a \in Q: a[i]=v \lor *$ then this case can only occur if $v=v'$.
  \item[Case] \textsc{Maybe} $v'$:\\
  This decision state can be reached in one of two ways:
  \begin{description}
    \item[Case] $p$ read $v'$ from register $i$ of some acceptor $a$:
    Since at most value is written to each round, this case can only occur if $v=v'$.
    \item[Case] $p$ read $v'$ from a register $>i$:
    Since $v'$ is the first value to be written to a register $>i$, this case cannot occur.
  \end{description}
  \item[Case] \textsc{None}:\\
  This decision state can be reached in one of two ways:
  \begin{description}
    \item[Case] $p$ read \emph{nil} from register $i$ of some acceptor $a \in Q$:\\
    Since $\forall a \in Q: a[i]=v \lor *$, this case cannot occur.
    \item[Case] $p$ read two different non-nil values from rounds $\geq i$:\\
    Since the proposer can only read $v$ from round $i$ and cannot have read any non-nil values from registers sets $>i$, this case cannot occur.
  \end{description}
\end{description}
Each case either requires that $v=v'$ or cannot occur, therefore it must be case that $v=v'$
\end{proof}

\begin{lemma}[Satisfying Rule 4 - Inductive case]
If the value $v$ is decided in round $i$ and all (non-nil) values written to registers $>i$ are $v$ then the next (non-nil) value to be written to a round $i'$ where $i<i'$ is also $v$.
\end{lemma}

Since the following proof overlaps significantly with the previous proof, we have underlined the parts which have been altered.

\begin{proof}
Assume the value $v$ is decided in round $i$ by quorum $Q \in \mathcal{Q}_i$ thus at some time $\forall a \in Q: a[i]=v$.
Since registers are write-once, it will always be that case that $\forall a \in Q: a[i]=v \lor *$.

Assume the value $v'$ is written to round $i'$ by proposer $p$ where $i<i'$.
\uline{Assume that all (non-nil) values written to registers $>i$ are $v$ thus $p$ can only read $v$ from (non-nil) registers $>i$.}
We will show that $v=v'$.

Consider the decision table of proposer $p$ when it is writing $v'$ to $i'$.
Since $i<i'$, the decision state of $Q$ must be either \textsc{None}, \textsc{Maybe}~$v'$ or \textsc{Decided}~$v'$ (Fig.~\ref{fig:decisiontable}).
We now consider each case in turn.
\begin{description}
  \item[Case] \textsc{Decided} $v'$:\\
  This decision state requires that $\forall a \in Q: a[i]=v'$.
  Since we know that $\forall a \in Q: a[i]=v \lor *$ then this case can only occur if $v=v'$.
  \item[Case] \textsc{Maybe} $v'$:\\
  This decision state can be reached in one of two ways:
  \begin{description}
    \item[Case] $p$ read $v'$ from register $i$ of some acceptor $a$: 
    Since at most value is written to each round, this case can only occur if $v=v'$.
    \item[Case] $p$ read $v'$ from a register $>i$:
    \uline{Since $v$ is the only (non-nil) value to be written to registers $>i$, then this case can only occur if $v=v'$.}
  \end{description}
  \item[Case] \textsc{None}:\\
  This decision state can be reached in one of two ways:
  \begin{description}
    \item[Case] $p$ read \emph{nil} from register $i$ of some acceptor $a \in Q$:\\
    Since $\forall a \in Q: a[i]=v \lor *$, this case cannot occur.
    \item[Case] $p$ read two different non-nil values from rounds $\geq i$:\\
    \uline{Since the proposer can only read $v$ from round $\geq i$, this case cannot occur.}
  \end{description}
\end{description}
Each case either requires that $v=v'$ or cannot occur, therefore it must be case that $v=v'$
\end{proof}

\end{document}